\documentclass{comnet}

\usepackage{amsmath,amsfonts}
\usepackage{algorithmic}
\usepackage{algorithm}
\usepackage{array}
\usepackage[caption=false,font=normalsize,labelfont=sf,textfont=sf]{subfig}
\usepackage{textcomp}
\usepackage{stfloats}
\usepackage{url}
\usepackage{verbatim}
\usepackage{graphicx}
\usepackage{graphics}
\usepackage{comment}
\usepackage{tabularx}
\usepackage{amssymb}
\usepackage{balance}
\usepackage{amsthm}

\begin{document}

\title{An Information Theory Approach to Network Evolution Models}

\shorttitle{An Information Theory Approach to Network Evolution Models} 
\shortauthorlist{A. Farzaneh and J. P. Coon} 

\author{
\name{Amirmohammad Farzaneh}
\address{Department of Engineering Science, University of Oxford, Oxford, UK}
\and
\name{Justin P. Coon$^*$}
\address{Department of Engineering Science, University of Oxford, Oxford, UK\email{$^*$Corresponding author: justin.coon@eng.ox.ac.uk}}}

\maketitle

\begin{abstract}
{A novel Markovian network evolution model is introduced and analysed by means of information theory. It will be proved that the model, called Network Evolution Chain, is a stationary and ergodic stochastic process. Therefore, the Asymptotic Equipartition Property can be applied to it. The model's entropy rate and typical sequences are also explored. Extracting particular information from the network and methods to simulate network evolution in the continuous time domain are discussed. Additionally, the Erdős–Rényi Network Evolution Chain is introduced as a subset of our model with the additional property of its stationary distribution matching the Erdős–Rényi random graph model. The stationary distributions of nodes and graphs are calculated for this subset alongside its entropy rate. The simulation results at the end of the paper back up the proved theorems and calculated values.}
{Network Evolution, Information Theory, Markov Chain, Stochastic Process, Entropy Rate}
\end{abstract}

\section{Introduction}
Graphs have often been used to model networks in different fields of science. Many random graph models have been introduced throughout the years to model different types of networks. Two of the oldest and the most important ones were introduced back in 1959 \cite{erdHos1960evolution}, \cite{gilbert1959}. Numerous works demonstrate the ability of such models to capture and simulate the properties of different networks, such as social networks \cite{ROBINS2007173}, electric networks \cite{holmgren2006using}, and even biological networks \cite{10.1093/bioinformatics/btm370}. However, we must take into account that networks are dynamic entities. Networks change and evolve over time, and a network model that captures the dynamics of the network is needed for better analysis of the network. The evolution of networks has been greatly studied throughout the years, and this has resulted in models that capture the growth of networks, such as the Barabási–Albert model \cite{dorogovtsev2002evolution}. There have also been models introduced as  Network Evolution Models, which simulate networks and their evolution over time \cite{doi:10.1080/0022250X.1996.9990179}, \cite{TOIVONEN2009240}. These models have proven to be very useful in simulating the behaviour of evolving networks, and have been able to answer many important questions in network science. However, our primary goal in this paper is to provide quantitative answers to questions such as how much we expect the network to evolve over time or how compressible the network is. This can be used in applications in which we need to predict the future states of the network, or transmit the states over a communication channel in an efficient manner. Moreover, most of the existing models only consider the growth of networks and do not consider the scenarios in which the network shrinks. The shrinkage of networks is observable in many scenarios, such as the departure of people from online social networks. In this paper, we introduce a Network Evolution Model that can be examined using an information theory point of view, and is powerful in capturing the dynamics of real-life networks. We believe that this model can easily simulate the dynamics of data networks, social networks, biological networks and any other form of network that can be represented by a graph.

The problem of finding the more probable sequences of a particular random variable is a fundamental principle of data compression in information theory. Typical sequences are defined as a result of the Asymptotic Equipartition Property (AEP), and are proved to essentially represent sequences that have a very high probability of occurrence. The concept of typical sequences motivates us to search for a similar approach in relation to graphs. There are many application scenarios in which the topology of the network needs to be transmitted and we need to compress it as much as possible. In many of these scenarios, we have the current state of the network and want to predict or compress the future states. This motivates us to define a network evolution model that both simulates the network well and can be analysed using information theory. Additionally, analysing the dynamics of graphs by means of information theory will result in calculating a value for the entropy rate of the process. \cite{zhao2011entropy} uses a similar approach, and is very successful in calculating the entropy rate of nonequilibrium growing networks, and then analysing them. However, there are differences between the interpretation of entropy in our model and the one introduced by \cite{zhao2011entropy}, which is further discussed in section \ref{entropy rate and typical}.

After stating the Shannon-McMillan-Breiman theorem, also known as the general AEP, we will define a model for building a sequence of graphs for which typical sequences can be defined. This model beautifully captures the essence in which networks evolve over time. Consequently, a typical sequence in this model can be interpreted as a likely journey that a network will go through. This new model will then be able to help us identify the most probable states of a dynamic network. This model will be called an NEC, short for Network Evolution Chain. In many cases, it is not the network itself we are interested in, but rather certain properties of the network. Therefore, following the definition of an NEC, Network Property Chains are introduced as a method to only extract information about the desired properties of the network. Subsequently, we will provide two methods to expand the definition of Network Evolution Chains into the continuous time domain. Afterwards, a subset of NEC will be presented as Erdős–Rényi Network Evolution Chains (ERNEC). As the Erdős–Rényi model is commonly used for simulating graphs and networks, we are also interested in having Network Evolution Chains whose stationary probability distribution matches that of the random graphs generated by this model. It will be proved that ERNECs satisfy our desired condition. Furthermore, the entropy rate and stationary distribution of an ERNEC is calculated. Finally, simulation results are presented to back up the theories that we prove about these chains.

\section{Basics and Definition}

In this section, we introduce a stochastic process that manages to capture and model the way networks change over time. In information theory, the Asymptotic Equipartition Property is the basis of defining typical sequences and compressing random variables. However, the sequence needs to fulfil certain conditions for the AEP to be applied to it. In order for our model to be as inclusive as possible, we use a more general version of the AEP which is known as the Shannon-McMillan-Breiman theorem.

\begin{theorem}[AEP: Shannon-McMillan-Breiman Theorem]\cite[Thm.~16.8.1]{10.5555/1146355}
\label{General AEP}
If \(H\) is the entropy rate of a finite-valued stationary ergodic process \(\{X_i\}\), then
\[ -\frac{1}{n}\log p(X_1,X_2,...,X_{n})\rightarrow H\quad with~probability~1.\]
\end{theorem}

We now define a stochastic process that matches our desired model and prove that it satisfies the conditions of theorem \ref{General AEP}. We can then define typical sequences of the model.

\begin{definition}[Network Evolution Chain]
\label{NEC}
A stochastic process \(\{G_i\}\) is called a Network Evolution Chain (NEC) if 
\begin{itemize}
     \item \(G_1\) is a graph with exactly 1 node.
     \item \(G_i\) (\(i>1\)) is a graph created from \(G_{i-1}\). To this end, a transition scheme is first chosen among the following options: deletion, addition, and same. The probability of choosing deletion, addition, or same is \(r(G_{i-1})\), \(t(G_{i-1})\), and \(s(G_{i-1})\), respectively. Therefore, we have \(r(G_{i-1})+t(G_{i-1})+s(G_{i-1})=1\). Based on the chosen transition scheme, one of the following procedures will be done for creating \(G_i\):
     \begin{itemize}
        \item \textbf{Deletion:} \(G_i\) is \(G_{i-1}\) with one of its nodes deleted according to a probability distribution.
        \item \textbf{Addition:}  \(G_i\) is generated by adding a node to \(G_{i-1}\). The edges between the new node and the existing ones are added according to a random graph generation model.
        \item \textbf{Same:} \(G_{i}\) is the same as \(G_{i-1}\).
     \end{itemize}
     \item There exists an integer \(n_{max}\geq 1\) such that if \(G_i\) is a graph with \(n_{max}\) nodes then we have \(t(G_i)=0\). For all other graphs \(t(G_i)\neq 0\).
     \item If \(G_i\) is a graph with only one node, then we have \(r(G_i)=0\). For all other graphs \(r(G_i)\neq 0\).
     \item For all graphs we have \(s(G_{i})\neq 0\).
   \end{itemize}
\end{definition}

Fig~\ref{fig1} illustrates four possible consecutive steps of an instance of an NEC defined based on definition \ref{NEC}.

\begin{figure}[!h]
\centering
\includegraphics[width = \columnwidth]{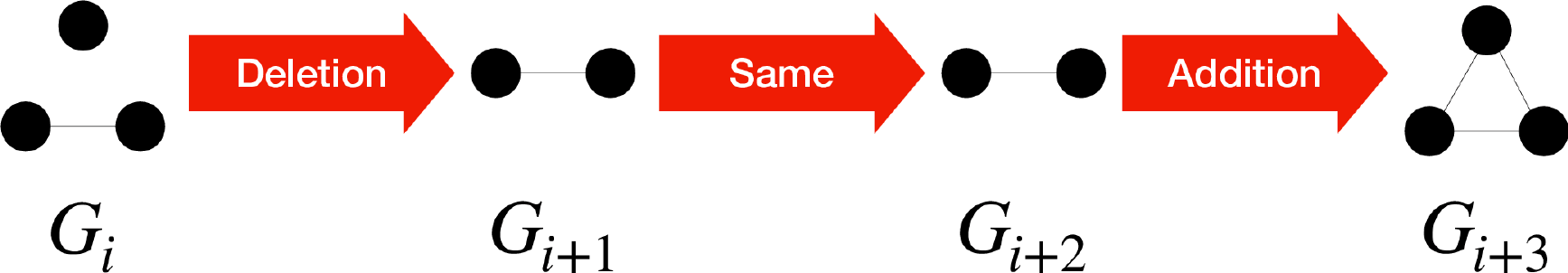}
\caption{\bf Example of four consecutive graphs in an NEC.}

\label{fig1}
\end{figure}

The following statements will aid us in proving that a Network Evolution Chain satisfies the conditions of theorem \ref{General AEP}.

\begin{definition}[Irreducible Markov chain]\cite[Ch.~4]{10.5555/1146355}
\label{Irreducible}
``If it is possible to go with positive probability from any state of a Markov chain to any other state in a finite number of steps, the Markov chain is said to be irreducible.''
\end{definition}

\begin{definition}[Aperiodic Markov chain]\cite[Ch.~4]{10.5555/1146355}
\label{Aperiodic}
``If the largest common factor of the lengths of different paths from a state to itself in a Markov chain is 1, the Markov chain is said to be aperiodic."
\end{definition}

\begin{lemma}\cite[Ch.~4]{10.5555/1146355}
\label{1}
``If the finite-state Markov chain \(\{X_i\}\) is irreducible and aperiodic, the stationary distribution is unique, and from any starting distribution, the distribution of \(X_n\) tends to the stationary distribution as \(n\rightarrow \infty\)."
\end{lemma}

\begin{corollary}\cite[Theorem~2]{pakes1969some}
\label{Markov Ergodicity}
A finite-state, irreducible, and aperiodic Mar\-k\-ov chain will form an ergodic process.
\end{corollary}
\begin{proof}
Theorem 2 from \cite{pakes1969some} provides a condition under which an irreducible and aperiodic Markov chain is ergodic. It can be easily shown that if the Markov chain is finite-state, it fulfils the specified condition in the mentioned theorem and the Markov chain will be ergodic.
\end{proof}

The following theorem states and proves that a Network Evolution Chain will eventually satisfy the conditions of theorem \ref{General AEP}.

\begin{theorem}
\label{NECstat}
A Network Evolution Chain tends towards a stationary and ergodic stochastic process as \(n\rightarrow \infty\).
\end{theorem}

\begin{proof}
Using lemma \ref{1} and corollary \ref{Markov Ergodicity}, it suffices for us to show that an NEC is a finite-state, irreducible and aperiodic Markov chain. We will talk about each of these in turn. Firstly, based on its definition, an NEC is in fact a Markov chain, as \(G_i\) is a random graph that only depends on \(G_{i-1}\). Furthermore, this Markov chain is also finite-state. Notice that the definition of an NEC does not allow us to have graphs with more than \(n_{max}\) nodes in the chain. Therefore, our state space consists of all graphs whose number of nodes is less than or equal to \(n_{max}\) and greater than or equal to 1, which is finite. An NEC is also irreducible. To prove this, we must present a manner in which we can go from any state to another. Suppose we want to go from graph \(g_1\) to graph \(g_2\), and both of these graphs are in the state space. If \(g_1=g_2\), we are done. Otherwise, we first remove nodes from \(g_1\) one by one until we reach a graph with only one node. Notice that this is possible because \(r(g)\neq 0\) for \(\{g|g\in S,|V(g)|\neq 1\}\) where \(S\) shows the state space of possible graphs and \(|V(g)|\) shows the number of nodes in graph $g$. Afterwards, we can add nodes to our one node graph one by one until we build \(g_2\). This is also possible because \(t(g)\neq 0\) for \(\{g|g\in S,|V(g)|\neq n_{max}\}\). To prove that an NEC is aperiodic, notice that the shortest path we can take from a graph to itself is of length 1. This is because for all \(g\in S\) we have \(s(g)\neq 0\). Since the largest common factor of any number and 1 is 1, we have proved that this chain is aperiodic. Now that we have proved that an NEC fulfils all the conditions in lemma \ref{1} and corollary \ref{Markov Ergodicity}, we can conclude that an NEC will tend towards a stationary and ergodic process as \(n\rightarrow \infty\).

\end{proof}

The NEC provides us with a model that generates a chain of random graphs. The name `Network Evolution Chain' is due to the fact that an NEC can beautifully describe the manner in which networks evolve over time. Every network initially begins with a single node, just like an NEC. Afterwards, at any point in time, a node can be added or deleted, or the network can stay the same as it was. Additionally, we may need to add or remove edges, but not nodes. These cases can be simulated by first deleting a node from one end of the edge we want to add or remove and then add that node again with our desired change in its edge connections. 

To model the evolution of a network using an NEC, there are three important factors that need to be chosen according to that specific network. The first one is the model according to which nodes and edges are added to the graph in case we are adding a node. The $G(n,p)$ model \cite{gilbert1959} and the Barabási–Albert model \cite{albert2002statistical} are examples of the models that can be used. The second factor is the probabilities of transition. In other words, \(r(g)\), \(t(g)\), and \(s(g)\) need to be defined for all \(g \in S\). There are many ways to set these probabilities. They can only depend on the number of nodes in \(g\) or can be unique to each graph. The final factor is the probability distribution according to which the deleted node is chosen in case of a deletion event. For example, the node to be deleted can be chosen uniformly among all the nodes in the graph, or nodes with fewer edges can have a higher probability of getting deleted. Undoubtedly, the best way to chose these protocols is according to the data observed from the behavior of the network we want to simulate.

\section{Entropy Rate and Typical Sequences}
\label{entropy rate and typical}

We have now defined a process called a Network Evolution Chain, which is both ergodic and stationary for large $n$. Assume the process \(\{G_i\}\) to be an NEC generated using definition \ref{NEC}. According to Theorem \ref{1}, this chain will tend towards its stationary distribution as $n$ grows large. Therefore, we can use theorem \ref{General AEP} and write

\begin{equation}
\label{H}
    -\frac{1}{n}\log p(G_1,\ldots,G_n)\rightarrow H(G)\quad with~probability~1.
\end{equation}

where \(H(G)\) shows the entropy rate of the stochastic process \(\{G_i\}\).

Based on equation \ref{H}, we can define the typical sets of an NEC.

\begin{definition} \cite[Ch.~3]{10.5555/1146355}
\label{Typical Set}
The typical set \(A^{(n)}_\epsilon\) with respect to a defined NEC is the set of all sequences \((G_1,G_2,...,G_n)\in S^n\) with the property
\[2^{-n(H(G)+\epsilon)}\leq p(G_1,G_2,...,G_n)\leq 2^{-n(H(G)-\epsilon)}.\]
\end{definition}

It can be proved that ``the typical set has probability nearly 1, all elements of the set are nearly equiprobable, and the number of elements in the set is nearly \(2^{nH}\)." \cite[Ch.~3]{10.5555/1146355}.

We now need to calculate the entropy rate of NEC. To this end, we use a result from \cite{10.5555/1146355}, which is stated in the following lemma.

\begin{lemma}\cite[Theorem~4.2.1]{10.5555/1146355}
\label{lemma entropy rate}
For a stationary stochastic process $\{X_i\}$, the following limit exists and is equal to the entropy rate of the process:
\[H(X)=\lim_{n \to \infty} H(X_n|X_{n-1},X_{n-2},...,X_1)\]
\end{lemma}

The entropy rate defined in lemma \ref{lemma entropy rate} can be interpreted in two ways. Firstly, based on its definition, it shows the conditional entropy of the last observed random variable given all the symbols that came before it. Secondly, \cite{10.5555/1146355} proves that the defined entropy rate also indicates the per symbol entropy of the $n$ observed random variables.

As an NEC is a stationary stochastic process for large $n$, we can use lemma \ref{lemma entropy rate} to calculate its entropy rate. We can write the following chain of equations.

\begin{subequations}
\begin{align}
\label{12}
H(G)=&\lim_{n \to \infty} H(G_n|G_{n-1},G_{n-2},...,G_1)\\
\label{13}
=&\lim_{n \to \infty}H(G_n|G_{n-1})\\
\label{extra}
&\begin{aligned}
=&\lim_{n \to \infty}H(G_n|G_{n-1},T)\\
&+H(T|G_{n-1})-H(T|G_{n-1},G_n)\\
\end{aligned}\\
\label{14}
=&\lim_{n \to \infty}H(G_n|G_{n-1},T)+H(T|G_{n-1})
\end{align}
\end{subequations}

Equality \ref{12} is the definition of entropy rate based on lemma \ref{lemma entropy rate}. Equality \ref{13} holds because of the fact that \(G_n\) is only dependent on \(G_{n-1}\). Equality \ref{extra} comes from taking into account the three states of addition, deletion, and remaining the same. $T$ is defined as a random variable that indicates which of the progression methods we have chosen for going from \(G_{n-1}\) to \(G_n\). In other words, we have \(T\in\{addition,deletion,same\}\). Finally, equality \ref{14} holds because having both \(G_{n-1}\) and \(G_n\), there will be no uncertainty about which transition scheme has been chosen for going from \(G_{n-1}\) to \(G_n\) and therefore we will have \(H(T|G_{n-1},G_n)=0\). For calculating \(H(T|G_{n-1})\), we can write
\begin{equation}
\label{H(T)}
H(T|G_{n-1})=-\sum_{g\in S}p(g)(t(g)\log t(g)+r(g)\log r(g)+s(g)\log s(g)).
\end{equation}
In equation \ref{H(T)}, \(p(g)\) shows the stationary distribution of graph \(g\). Now that \(H(T|G_{n-1})\) has been calculated, we move on to calculate \(H(G_n|G_{n-1},T)\).

\begin{subequations}
\begin{align}
\label{15}
\begin{split}
    H(G_n|G_{n-1},T)=&\sum_{g\in S}p(g)t(g)H(G_n|G_{n-1}=g,T=addition)\\
    &+\sum_{g\in S}p(g)r(g)H(G_n|G_{n-1}=g,T=deletion)\\
    &+\sum_{g\in S}p(g)s(g)H(G_n|G_{n-1}=g,T=same)
\end{split}\\
\label{16}
\begin{split}
=&\sum_{g\in S}p(g)t(g)H(G_n|G_{n-1}=g,T=addition)\\
    &+\sum_{g\in S}p(g)r(g)H(G_n|G_{n-1}=g,T=deletion)
\end{split}
\end{align}
\end{subequations}

Equality \ref{15} comes from conditional entropy and equality \ref{16} comes from the fact that we have \(H(G_n|G_{n-1}=g,T=same)=0\) because there remains no uncertainty when a graph stays the same as it was before. Notice that equation \ref{16} can not be more simplified because the node addition and deletion models are not known.

We can now provide the following equation for calculating the entropy rate of a Network Evolution Chain.
\begin{multline}
    \label{entropy rate calculation}
        H(G) = -\sum_{g\in S}p(g)(t(g)\log t(g)+r(g)\log r(g)+s(g)\log s(g))\\
         +\sum_{g\in S}p(g)t(g)H(G_n|G_{n-1}=g,T=addition)\\
    +\sum_{g\in S}p(g)r(g)H(G_n|G_{n-1}=g,T=deletion)
\end{multline}
We observe that calculating the exact value of the entropy rate depends on the factors below:
\begin{enumerate}
    \item \(t(g)\), \(r(g)\), and \(s(g)\) for all \(g\in S\).
    \item Model of choice for adding a node and its edges to \(g\) for all \(g\in S\).
    \item Model of choice for deleting a node from \(g\) for all \(g\in S\).
    \item The stationary probability distributions of all \(g\in S\).
\end{enumerate}
After knowing these factors about a Network Evolution Chain, the entropy rate for the chain can be calculated using equation \ref{entropy rate calculation}.

It is notable to mention the differences between the entropy rate calculated in this section with that of \cite{zhao2011entropy}. Firstly, \cite{zhao2011entropy} does not take into account the shrinkage of a network and the entropy is calculated asymptotically as the network grows large. Secondly, \cite{zhao2011entropy} has taken into account only two possible ways of connecting nodes: one based on nodes having a constant mean degree, and the other one based on the network being a tree. This is in contrast to an NEC, which can simulate any type of connection between nodes. Finally, the network size in \cite{zhao2011entropy} is not bounded, whereas in an NEC the network size is limited to $n_{max}$. All in all, these two approaches are different in their own ways and each one must be utilized in its appropriate application scenario.


\section{Functions of Network Evolution Chains}
\label{Functions of NEC}

We have introduced Network Evolution Chains in order to simulate the way networks evolve over time. When a certain network is simulated using an NEC, we have the whole graph of the network at each step as our random variable and there will be no information loss. However, if we want to only look at certain properties of the network, it may not be efficient to store the whole graphs in our chain. In this section, we introduce functions of Network Evolution Chains as a method for extracting specific information from the networks in the chain. To this end, we define Network Property Chains.

\begin{definition}[Network Property Chain (NPC)]
Consider the Network Evolution Chain \(G=\{G_1,G_2,...,G_n\}\). We define \(Y=\{Y_1,Y_2,...,Y_n\}\) as a Network Property Chain of \(G\) if there exists a function \(f\) so that we have
\[\forall i \in \{1,2,...,n\}\qquad Y_i=f(G_i).\]

\end{definition}

It can be seen that an NPC is defined using an NEC and a function \(f\). We can use Network Property Chains to look for certain properties and behaviors of a network. Specifically, we choose $f$ so that it returns the property of its input that we are interested in.

Notice that a function of a Markov chain is not necessarily a Markov chain. However, the following theorem will aid us in analysing Network Property Chains without needing them to be Markov chains.

\begin{theorem}
\label{ergodic stationary NPC}
A Network Property Chain is a stationary and ergodic stochastic process.
\end{theorem}
\begin{proof}
We first prove that an NPC is a stationary process. Consider \(P_i\) as the set of all \(g\in S\) for which we have \(f(g)=Y_i\). We can calculate the probability \(p(Y_i,...,Y_{i+k})\) using the following equation.

\begin{equation}
    \label{brute-force}
    \begin{split}
    p(Y_i,Y_{i+1},...,Y_{i+k})&=p(Y_i)\prod_{j=i+1}^{i+k}p(Y_j|Y_{j-1},...,Y_{i})\\
    &=p(P_i)\prod_{j=i+1}^{i+k}\sum_{g_1\in P_{j-1}}\sum_{g_2\in P_j}p(g_2|g_1)
    \end{split}
\end{equation}

As we know that an NEC is a stationary process, there are no terms in equation \ref{brute-force} which are time-variant and therefore an NPC is also stationary.

To prove that an NPC is ergodic, we use Birkhoff's ergodic theorem \cite{birkhoff1931proof}. For using Birkhoff's theorem to prove ergodicity, we need to show that all functions of an NPC satisfy Birkhoff's condition. Firstly, all functions of an NPC are also a function of its respective NEC. To see this, notice that if the function $f$ which is used for building the NPC is injective, then this statement is easily concluded and if $f$ is surjective, then any function of the NPC can be seen as a function of its respective NEC where the function takes on equal values for all $g_i$ and $g_j$ for which we have $f(g_i)=f(g_j)$. Additionally, we know that an NEC is in fact ergodic and all functions of it satisfy the condition of Birkhoff's theorem. Therefore, an NPC is also an ergodic process.

\end{proof}

Theorem \ref{ergodic stationary NPC} shows that Network Property Chains are stationary and ergodic processeses. Therefore, we can apply the Asymptotic Equipartition Property to these chains as well. We will end up with the following equation after applying the AEP to NPC.

\begin{equation}
    \label{AEP to NPC}
    -\frac{1}{n}\log p(Y_1,Y_2,...,Y_n)\to H(Y)\quad \text{as} \quad n\to \infty
\end{equation}

Applying the AEP to an NPC gives us the opportunity to define typical sequences for Network Property Chains. This means that we will be able to obtain the typical ways that properties of the network will evolve over time. For example, assume that we are interested in the number of triangles in the network. We will define \(f(g),g\in S\) so that it returns the number of triangles in \(g\). The typical set of \(\{Y_i\}\) will then contain the most probable evolution chains of the number of triangles in the network. 

Notice that as an NPC is not necessarily a Markov chain, calculating its entropy rate can not be done in the same way as NECs. However, we will be able to estimate its entropy rate using the AEP. To this end, we first need \(p(Y_1,Y_2,...,Y_n)\), which can be calculated using equation \ref{brute-force}. However, equation \ref{brute-force} may be very difficult to compute, particularly for NECs with a large \(n_{max}\). To assist us with the computation, we use the concept of a Hidden Markov Model (HMM) \cite{rabiner1986introduction}. HMMs are very useful for simulating Network Property Chains. As mentioned earlier, one of the goals for introducing NPCs was to remove the need for storing the whole graphs. In an HMM, the assumption is that the original Markov chain is not available (it is hidden), and we can only observe another chain that is based on our original Markov chain. The transition probabilities for the Hidden Markov Model \((G_n,Y_n)\) are the same as the transition probabilities for the NEC \(\{G_i\}\). Assume the set \(P\) to be the image of \(S\) under $f$. For the emission probabilities of the HMM for all \(g\in S\) and \(y\in P\) we have
\[
    p(y|g)= 
\begin{cases}
    1,& \text{if } f(g)=y\\
    0,              & \text{otherwise}
\end{cases}.
\]

Now that we have defined our Hidden Markov Model, we can make use of its properties to efficiently calculate the probability of the observed sequence. HMMs provide us with the forward algorithm \cite{rabiner1989tutorial} which is used to efficiently calculate the probability of an observed sequence. Therefore, our solution for calculating \(p(Y_1,Y_2,...,Y_n)\) is to first define the Hidden Markov Model $(G_n,Y_n)$ and then use the forward algorithm. This will provide us with the sufficient tool to calculate the probability of observed sequences, estimate the entropy rate of the Network Property Chain, and also recognize typical sequences.


\section{Continuous Time Network Evolution Chains}

We have observed that Network Evolution Chains have the ability to simulate the dynamics of networks. However, they fail to provide us with a simulation for the network in continuous time. In this section, we provide two methods to overcome this limitation.

\subsection{The \(s(g)\) Method} \label{Continous}

The first method uses the value of \(s(g)\)(\(g\in S\)) to simulate the continuous time domain. The idea is that we first quantize time and consider the time between each step of the chain to be a constant \(\tau\). Now, for each \(g\in S\), we can set \(s(g)\) according to how much time we expect $g$ to remain unchanged. Consider \(L(g)\) to be the random variable that shows the time that graph \(g\) stays unchanged and \(Q(g)\) to be the random variable that shows the number of consecutive steps that staying the same is chosen for graph \(g\). We can write

\begin{equation}
\begin{split}
    \mathbb{E}[L(g)]&=\tau(\mathbb{E}[Q(g)]+1)\\
    &=\tau (1+\sum_{i=0}^{\infty}i(s(g))^{i}(1-s(g)))\\
    &=\frac{\tau}{1-s(g)}.
    \label{psame method}
\end{split}
\end{equation}

After choosing our desired value for the time quantization variable \(\tau\), we can use equation \ref{psame method} to set the values for \(s(g)\) for all \(g\in S\). This way, the expected time that each network stays unchanged will match our desired value.

\subsection{The Continuous Time Markov Chain Method}

Continuous Time Markov Chains (CTMC) \cite[Ch.~5]{ross1996stochastic} are an extension to Markov chains with the goal to simulate Markov chains in the continuous time domain. Therefore, we can use Continuous Time Markov Chains to simulate an NEC in continuous time.

\begin{definition}
A Continuous Time Network Evolution Chain is a Continuous Time Markov Chain $G(t)$ in which the states and the transition probabilities are the same as a Network Evolution Chain and additionally, \(\lambda(g)\) is chosen for all \(g\in S\) as the parameter for the exponential distribution used for determining the time we stay at state \(g\).
\end{definition}

We can see that the states are generated in the same manner as explained for the discrete time NEC. Furthermore, we remain in \(G_i\) for \(\tau_i\) seconds. \(\tau_i\) is generated using an exponential distribution with parameter \(\lambda(G_i)\). We therefore have to choose a value \(\lambda(g)\) for all \(g\in S\). This value is what determines how much time we wait in each state according to an exponential distribution with its expected value equal to $1/\lambda(g)$. 

Suppose that we label the graphs in $S$ from 1 to $|S|$ according to any ordering. Let \(p_{ij}\) denote the probability of going from state $i$ to $j$ in the respective NEC. We use the following notation from chapter 5 of \cite{ross1996stochastic} to provide a result about the stationary distribution of Continuous Time Network Evolution Chains.
\[p_{ij}(t)=p(G(t+s)=j|G(s)=i)\]

\[
    r_{ij}= 
\begin{cases}
    \lambda(i)p_{ij},& \text{if } i\neq j\\
    -\lambda(i),& \text{if } i=j\\
\end{cases}
\]

$p_{ij}$ represents the probability that the process goes to state $j$ when it leaves state $i$. Based on the definition, \(p_{ij}(t)\) represents the probability of the process with current state of $i$ moves to state $j$ after time $t$ has passed, and \(r_{ij}\) represents the rate of going from state $i$ to state $j$. To assist us with the calculations, the matrices $\textbf{\textit{R}}$, $\textbf{\textit{P}}(t)$, and $\textbf{\textit{P}}^\prime(t)$ are defined. The element in row $i$ and column $j$ of these matrices are, respectively, $r_{ij}$,  $p_{ij}(t)$, and $p_{ij}^\prime(t)$ \cite[Ch.~5]{ross1996stochastic}. It is stated in \cite{ross1996stochastic} that the following equation is the solution for finding the stationary distribution of a CTMC.

\begin{equation}
    \label{Stationary Solution CTMC}
    \textbf{\textit{P}}(t)=e^{\textbf{\textit{R}}t} = \sum_{t=0}^\infty (\textbf{\textit{R}}t)^i/i!
\end{equation}
where \(\textbf{\textit{R}}^0\) is the identity matrix.

Equation \ref{Stationary Solution CTMC} provides us with the solution to find the stationary continuous time distributions. However, it may be computationally infeasible to calculate the distribution using this equation. To overcome this challenge, \cite{ross1996stochastic} presents a method to estimate the answer to equation \ref{Stationary Solution CTMC} \cite[pg.~250]{ross1996stochastic}.

\section{Erdős–Rényi Network Evolution Chains}
The \(G(n,p)\) model \cite{gilbert1959}, often called the Erdős–Rényi model, is a mathematical model used in graph theory for generating random graphs. In this paper, we use the terms \(G(n,p)\) and Erdős–Rényi interchangeably. A \(G(n,p)\) graph is a graph with \(n\) nodes, and each possible edge in this graph is present independently from others with probability \(p\). \(G(n,p)\) graphs can be generated by starting with only one node, and then adding other nodes to the graph one by one. Each new node that is added gets connected to each of the existing nodes with probability \(p\), independent from others.

If we create a Network Evolution Chain whose addition step is based on \(G(n,p)\), we are interested for the stationary conditional probability of observing a graph $g$ given its number of nodes \(|V(g)|\) to match the probability of the same graph if it was generated by the \(G(|V(g)|,p)\) model. Providing a necessary and sufficient condition for an NEC to achieve this property is challenging. Instead, we introduce a subset of NEC as Erdős–Rényi Network Evolution Chains and then prove that it satisfies our desired condition.

We will use the notation \(S_n\) to show the subset of graphs in \(S\) that have \(n\) nodes. Furthermore, to avoid confusion in calculations in which \(p\) is used to show the probability of an event, we will use \(q\) instead of \(p\) as the parameter for the Erdős–Rényi model.

\begin{definition}[Erdős–Rényi Network Evolution Chain (ERNEC)]
\label{ERNEC Definition}
We define an Erdős–Rényi Network Evolution Chain as a Network Evolution Chain with the following additional properties:

\begin{itemize}
    \item for all \(1\leq n\leq n_{max}\) and \(g_1,g_2\in S_n\) we have:
    \begin{itemize}
        \item \(t(g_1)=t(g_2)\)
        \item \(r(g_1)=r(g_2)\)
        \item \(s(g_1)=s(g_2)\)
    \end{itemize}
    \item The node in the first graph of the chain is labeled as 1.
    \item If at any point in time, addition is chosen, a new node with its label equal to the highest existing label plus one gets added to the graph. The edges between the new node and the existing nodes are added using the Erdős–Rényi model.
    \item If at any point in time, deletion is chosen, one of the existing nodes gets deleted according to a uniform distribution on the existing nodes. Consider the label of the deleted node to be \(i\). After the deletion of node \(i\), the labels of all existing nodes whose label is greater than \(i\) are reduced by one.
\end{itemize}
\end{definition}

We now move on to prove that an ERNEC has the \(G(n,q)\) distribution when conditioned on its number of nodes. We first prove the following lemma.

\begin{lemma}
\label{dependence}
For an Erdős–Rényi Network Evolution Chain created based on the \(G(n,q)\) model, The stationary distribution \(p(g)\) for \(g\in S_n,n\geq2\) is in the form of \(Cq^k(1-q)^{(n-1-k)}p(g^\prime)\), with \(g^\prime\) being \(g\) without node \(n\) (its newest node), C being a constant value among all \(g\in S_n\) and k being the number of nodes in \(g^\prime\) to which node \(n\) is connected.
\end{lemma}
\begin{proof}
We will prove this lemma using backward induction on the number of nodes in the graph. Therefore, we have to start with graphs with \(n_{max}\) nodes. For \(g \in S_{n_{max}}\) we can write
\begin{align}
p(g)&=s(g)p(g)+p(g^\prime)t(g^\prime)q^{k}(1-q)^{n_{max}-1-k}\\
\label{g}
&=\frac{p(g^\prime)t(g^\prime)q^{k}(1-q)^{n_{max}-1-k}}{1-s(g)}.
\end{align}

In equation \ref{g}, \(g^\prime\) is the graph \(g\) after its node with label \(n_{max}\) is deleted from the graph and $k$ is the number of edges connected to node \(n_{max}\). We can see that as \(t(g^\prime)/(1-s(g))\) is a constant value for all \(g\in G_{n_{max}}\), our base case is proved.

We will now move on to our induction step. Suppose that we know that the statement holds for all \(S_k\) where \(k>n\). We now have to prove the case for \(S_n\). For all \(g\in S_n\) we can write
\begin{equation}
\label{eq1}
\begin{split}
p(g)=s(g)p(g)+t(g^\prime)p(g^\prime)q^k(1-q)^{n-1-k}\\
+r(g^{+})\sum_{i=1}^{n+1}\frac{1}{n+1}\sum_{j=1}^{2^{n}}p(g^{+}_{ij}).
\end{split}
\end{equation}
In equation \ref{eq1}, the first term accounts for the case where the network is created using a $same$ step and the next two account for the $addition$ and $deletion$ cases respectively. The first summation in the $deletion$ term is due to the fact that the deleted node could have been any of the \(n+1\) nodes in the previous graph with probability \(1/(n+1)\). The second summation is to account for all possible ways that the deleted node could have been connected to the remaining nodes. \(g^{+}_{ij}\) indicates the $j$th possible way that edges could have been connected to the deleted node which was labeled as \(i\). \(r(g^{+})\) shows \(r(g^{+}_i)\) for all $i$, as they are all equal. 

Notice that the deleted node could have been any of the \(n+1\) nodes in the previous graph with equal probability, and no matter its label, it could have been connected to the existing \(n\) nodes in any manner with its respective Erdős–Rényi probability. Because of this symmetry in the deletion step between different nodes for all \(g\in S_{n+1}\), we can simplify the $deletion$ term in the following manner.
\begin{equation}
\label{deletion}
\begin{split}
&r(g^{+})\sum_{i=1}^{n+1}\frac{1}{n+1}\sum_{j=1}^{2^{n}}p(g^{+}_{ij})\\
&=r(g^{+})(\sum_{j=1}^{2^{n}}p(g^{+}_{(n+1)j}))(\sum_{i=1}^{n+1}\frac{1}{n+1})\\
&=r(g^{+})\sum_{j=1}^{2^{n}}p(g^{+}_{(n+1)j})\\
&=r(g^{+})\sum_{j=1}^{2^{n}}p(g^{+}_{j})
\end{split}
\end{equation}
where the notation in the last equality is changed for simplicity.

By induction, we can write \(p(g^{+}_j)\) in the following manner.
\[p(g^{+}_j)=C^{+}p(g)q^{k^{+}_j}(1-q)^{n-k^{+}_j}\]
and because in \(\sum_{j=1}^{2^{n}}p(g^{+}_{j})\) we are including all possible ways that edges can get connected to node \(n+1\), we can further simplify equation \ref{deletion}.
\begin{equation}
\label{deletion2}
\begin{split}
r(g^{+})\sum_{j=1}^{2^{n}}p(g^{+}_j)&=r(g^{+})C^{+}p(g)\sum_{j=1}^{2^{n}}q^{k^{+}_j}(1-q)^{n-k^{+}_j}\\
&=r(g^{+})C^{+}p(g)
\end{split}
\end{equation}
By inserting equation \ref{deletion2} into equation \ref{eq1} we will have
\[(1-s(g)-r(g^{+})C^{+})p(g)=t(g^\prime)p(g^\prime)q^k(1-q)^{n-1-k}.\]
And therefore we end up with the following equation.
\begin{equation}
\label{P(g)}
\begin{split}
    p(g)&=\frac{t(g^\prime)}{1-s(g)-r(g^{+})C^{+}}q^k(1-q)^{n-1-k}p(g^\prime)\\
    &=Cq^k(1-q)^{n-1-k}p(g^\prime)
\end{split}
\end{equation}
Notice that \(t(g^\prime)\), \(s(g)\), \(r(g^{+})\), and \(C^{+}\) have a constant value for all \(g\in S_n\). Therefore, the value \((t(g^\prime))/(1-s(g)-C^{+})\) can be considered as a constant \(C\) for all graphs in \(S_n\). Our lemma is therefore proved by backward induction.
\end{proof}

We are now equipped to prove the following theorem.
\begin{theorem}
\label{proof ER}
The conditional distribution of graphs in an Erdős–Rényi Network Evolution Chain when conditioned on their number of vertices is equivalent to the same Erdős–Rényi distribution used in the addition step of the NEC.
\end{theorem}
\begin{proof}
We will use induction to prove this theorem. It is easy to see that as there exists only one graph with one node, its conditional distribution matches that of the Erdős–Rényi model. This will act as the basis of our induction.

Now, suppose that the theorem holds for the subset \(\{g\in S||V(g)|<n\}\). For the induction step, we need to prove that the theorem holds for all \(g\in S_n\). We can write
\[p(g||V(g)|=n)=\frac{p(g,|V(g)|=n)}{p(|V(g)|=n)}=\frac{p(g)}{p(|V(g)|=n)}.\]
Therefore, if the probability ratio for all possible pairs of graphs from this subset are equal to the same ratio from the Erdős–Rényi distribution, we can conclude that the theorem holds for all \(g\in  S_n\). For \(g_1,g_2\in S_n\), we can write from lemma \ref{dependence}
\[\frac{p(g_1)}{p(g_2)}=\frac{Cq^{k_1}(1-q)^{(n-1-k_1)}p(g^\prime_1)}{Cq^{k_2}(1-q)^{(n-1-k_2)}p(g^\prime_2)}.\]
let \(E_1\), \(E_2\), \(E^\prime_1\) and \(E^\prime_2\) be the number of edges in \(g_1\), \(g_2\), \(g^\prime_1\) and \(g^\prime_2\) respectively. From induction, we know that
\[\frac{p(g_1^\prime)}{p(g_2^\prime)}=\frac{q^{E^\prime_1}(1-q)^{{n-1\choose 2}-E^\prime_1}}{q^{E^\prime_2}(1-q)^{{n-1\choose 2}-E^\prime_2}},\]
and therefore we can write
\[\frac{p(g_1)}{p(g_2)}=\frac{q^{k_1}(1-q)^{(n-1-k_1)}q^{E^\prime_1}(1-q)^{{n-1\choose 2}-E^\prime_1}}{q^{k_2}(1-q)^{(n-1-k_2)}q^{E^\prime_2}(1-q)^{{n-1\choose 2}-E^\prime_2}}\]
\[=\frac{q^{k_1+E^\prime_1}(1-q)^{{n-1\choose 2}+(n-1)-E^\prime_1-k_1}}{q^{k_2+E^\prime_2}(1-q)^{{n-1\choose 2}+(n-1)-E^\prime_2-k_2}}.\]
We can rewrite the equation in the following format.
\begin{equation}
\label{Chain is ER}
   \frac{p(g_1)}{p(g_2)} =\frac{q^{E_1}(1-q)^{{n\choose 2}-E_1}}{q^{E_2}(1-q)^{{n\choose 2}-E_2}}
\end{equation}
As equation \ref{Chain is ER} holds for all \(g_1,g_2\in S_n\), it proves that the conditional distribution of graphs in \(S_n\) matches that of the Erdős–Rényi model used in the addition step and our proof is complete.
\end{proof}

Now that we have proven that Erdős–Rényi Network Evolution Chains satisfy our desired property, we will move on to examine their other features.

\subsection{Stationary Distribution of the Number of Nodes}

One interesting feature to examine in Network Evolution Chains is the stationary distribution of the number of nodes of the graphs present in the chain. We know that such a distribution exists and is unique, because theorem \ref{ergodic stationary NPC} guarantees it (Notice that the chain containing the number of nodes for each element of an NEC is actually an NPC). The definition of an ERNEC makes the analysis of the stationary distribution of the number of nodes much easier. The following matrix can be defined as the transition probability matrix between different numbers of nodes in the ERNEC.
\begin{equation}
\label{Matrice}
\textbf{\textit{P}}=\begin{bmatrix}
s(1) & t(1)   & 0 & ... & 0&0\\
r(2)&s(2)&t(2)&...&0&0\\
0& r(3)&s(3)&...&0&0\\
\vdots&\vdots&\vdots&...&\vdots&\vdots\\
0&0&0&...&r(n_{max})&s(n_{max})
\end{bmatrix}
\end{equation}
In equation \ref{Matrice}, \(t(i)\), \(s(i)\) and \(r(i)\) show their respective values for any graph in the state space whose number of nodes is equal to \(i\). Let $\pi_i$ be the stationary probability of observing a graph with $i$ nodes. For finding the stationary distribution \(\pi=(\pi_1,\pi_2,...\pi_{n_{max}})\), we need to solve the equation \(\pi=\pi \textbf{\textit{P}}\) or equivalently, \((1-\textbf{\textit{P}})\pi=0\) constrained on \(\sum_{i=1}^{n_{max}}\pi_i=1\). Notice that based on our constraint, if we find \(\pi_i/\pi_1\) for all \(i\), we will have found  \((\pi_1,\pi_2,...\pi_{n_{max}})\). After some substituting \(s(i)\) with \(1-t(i)-r(i)\) in equations present in \((1-\textbf{\textit{P}})\pi=0\), we end up with the following system of equations.
\begin{equation}
\label{series}
    \begin{cases}
        \pi_1t(1)-\pi_2r(2)=0\\
        -\pi_1t(1)+\pi_2(r(2)+t(2))-\pi_3r(3)=0\\
        \vdots\\
        -\pi_{n_{max}-1}t(n_{max}-1)+\pi_{n_{max}}r(n_{max})=0
    \end{cases}
\end{equation}
We now provide the following theorem for calculating \(\frac{\pi_i}{\pi_1}\) for all \(i\).
\begin{theorem}
\label{stationary}
For all \(1<i\leq n_{max}\) we have
\[\frac{\pi_i}{\pi_1}=\frac{\prod_{j=1}^{i-1}t(j)}{\prod_{j=2}^{i}r(j)}.\]
\end{theorem}
\begin{proof}
We will prove this theorem by induction. For our base case, we use the first equation in \ref{series} and we have
\[\frac{\pi_2}{\pi_1}=\frac{t(1)}{r(2)},\]
which matches the theorem. Now, suppose that we know that the theorem holds for \(2\leq i \leq k\). For \(k+1\), we use the $k$th equation in \ref{series}. We have
\[-\pi_{k-1}t(k-1)+\pi_k(t(k)+r(k))-\pi_{k+1}r(k+1)=0.\]
If we divide the equation by \(\pi_{k-1}\), we end up with the following equation.
\begin{equation}
\label{frac}
    \frac{\pi_{k+1}}{\pi_{k-1}}=\frac{-t(k-1)+\frac{\pi_k}{\pi_{k-1}}(t(k)+r(k))}{r(k+1)}
\end{equation}
Notice that from the assumption of the induction, we can write
\[\frac{\pi_k}{\pi_{k-1}}=\frac{t(k-1)}{r(k)}.\]
By plugging this into equation \ref{frac} we have
\[ \frac{\pi_{k+1}}{\pi_{k-1}}=\frac{\frac{t(k)t(k-1)}{r(k)}}{r(k+1)}
=\frac{t(k)t(k-1)}{r(k)r(k+1)}.\]
And because we have \(\frac{\pi_{k-1}}{\pi_1}=\frac{\prod_{j=1}^{k-2}t(j)}{\prod_{j=2}^{k-1}r(j)}\) we will have
\[\frac{\pi_{k+1}}{\pi_1}=\frac{\prod_{j=1}^{k}t(j)}{\prod_{j=2}^{k+1}r(j)}.\]
\end{proof}
Based on theorem \ref{stationary}, for finding the stationary distribution of the number of nodes, we can choose \(\pi_{1,temp}=1\), find \(\pi_{i,temp}\) for \(2\leq i \leq n_{max}\) using the theorem and then divide \((\pi_{1,temp},...,\pi_{n_{max},temp})\) by \(\sum_{i=1}^{n_{max}}\pi_{i,temp}\). It is interesting to observe that the distribution does not depend on the values set for \(s(i)\). Therefore, for reaching our desired distribution on the number of nodes, we can set \(t(i)\)'s and \(r(i)\)'s accordingly. Theorem \ref{stationary} also provides us with the following corollary.
\begin{corollary}
For all \(1\leq l<m\leq n_{max}\) we have
\[\frac{\pi_m}{\pi_l}=\frac{\prod_{i=l}^{m-1}t(i)}{\prod_{i=l+1}^{m}r(i)}.\]
\end{corollary}

It is notable to mention that the model used in this section and the proved results can be seen as an analogue of the model and the results for birth-death processes \cite{karlin1957differential}.

\subsection{Probability Distribution of Individual Graphs}
We have proved that the conditional distribution of graphs in an Erdős–Rényi Network Evolution Chain matches the Erdős–Rényi model. We have also provided a method for calculating the stationary distribution of the number of nodes. Having both of these values for all \(g\in S\), we can calculate the probability of observing each individual graph. Consider \(g\in S\) with \(N\) vertices and \(E\) edges. For $N>1$ we can write
\begin{equation}
\label{gdist}
    p(g) = p(g||V(g)|=N)p(|V(g)|=N) = \pi_Nq^E(1-q)^{{N\choose 2}-E},
\end{equation}
where \(q\) is the parameter of the desired Erdős–Rényi model. For $N=1$ we simply have
\begin{equation}
    p(g) = p(g||V(g)|=1)p(|V(g)|=1)=\pi_1.
\end{equation}

Notice that as \(p(|V(g)|=N)\) does not depend on \(s(i)\), \(p(g)\) does not depend on \(s(i)\) either. This gives us the freedom to choose \(s(i)\) to match our desired continuous time model as explained in section \ref{Continous}.

\subsection{Entropy Rate}
For calculating the entropy rate of Erdős–Rényi Network Evolution Chains, we can use equation \ref{entropy rate calculation}. We first start by simplifying equation \ref{H(T)} for an ERNEC.
\begin{equation}
\label{insert1}
        H(T|G_{n-1})=-\sum_{i=1}^{n_{max}}\pi_i(t(i)\log t(i)+r(i)\log r(i)+s(i)\log s(i))
\end{equation}
We then move on to simplify equation \ref{16}. Notice that for Erdős–Rényi Network Evolution Chains we have \(H(G_n|G_{n-1}=g,T=deletion)=\log |V(g)|\). This is because in case of deletion, the deleted node is chosen uniformly among all the nodes of the graph. Therefore, we can write

\begin{equation}
\label{insert2}
    H(G_n|G_{n-1},T)=-\sum_{i=1}^{n_{max}-1}\pi_it(i)\sum_{j=0}^{i}{i\choose j}q^j(1-q)^{i-j}\log(q^j(1-q)^{i-j})+\sum_{i=2}^{n_{max}}\pi_ir(i)\log i.
\end{equation}

We can now calculate the entropy rate of ERNEC by inserting the results from equations \ref{insert1} and \ref{insert2} into equation \ref{entropy rate calculation}.

\begin{multline}
    \label{entropy rate}
        H(G)=-\sum_{i=1}^{n_{max}}\pi_i(t(i)\log t(i)+r(i)\log r(i)+s(i)\log s(i))\\
        -\sum_{i=1}^{n_{max}-1}\pi_it(i)\sum_{j=0}^{i}{i\choose j}q^j(1-q)^{i-j}\log(q^j(1-q)^{i-j})+\sum_{i=2}^{n_{max}}\pi_ir(i)\log i
\end{multline}

\subsection{Functions of Erdős–Rényi Network Evolution Chains}
\label{ER HMM}

In section \ref{Functions of NEC}, we defined a method that can be used for analyzing the dynamics of a function of the network using a Hidden Markov Model. Notice that the number of states in the proposed Hidden Markov Model equals the number of graphs that can be made with a maximum of \(n_{max}\) nodes. This number grows with \(O(e^{n_{max}^2})\). To reduce the number of states in the HMM proposed in section \ref{Functions of NEC}, we propose a new HMM for an ERNEC, which significantly reduces the number of states.

Consider \(G=\{G_1,G_2,...,G_n\}\) to be an Erdős–Rényi Network Evolution Chain. Also, consider the chain \(X=\{X_1,X_2,...,X_n\}\) to be a chain in which \(X_i\) shows the number of nodes in \(G_i\). The new chain, \(\{X_i\}\), is also a Markov chain. This is because the state of \(X_i\) only depends on \(X_{i-1}\) and its respective node transition probabilities. Specifically, the transition probabilities for \(\{X_i\}\) can be written as the following.

\[
    p(X_i|X_{i-1})= 
\begin{cases}
    t(X_{i-1}),& \text{if } X_{i}=X_{i-1}+1\\
    r(X_{i-1}),& \text{if } X_{i}=X_{i-1}-1\\
    s(X_{i-1}),& \text{if } X_{i}=X_{i-1}\\
    0, &\text{otherwise}\\
\end{cases}
\]

Now, consider another chain, \(Y=\{Y_1,Y_2,...,Y_n\}\). In this new chain, we have \(Y_i=f(G_i)\) for a desired function \(f\). We can build our HMM by using \(\{X_i\}\) as the original Markov chain and \(\{Y_i\}\) as the observed sequence. This way, emission probabilities can be written as

\[p(Y_i|X_i)=\frac{\sum_{g\in S_{X_i}}I(f(g)=Y_i)}{2^{{X_i \choose 2}}},\]
where \(I\) is a function which returns \(1\) if the condition inside it is satisfied and \(0\) otherwise.

This way, the number of states of the Markov chain will increase linearly with \(n_{max}\) and not exponentially. This will make the implementation of the HMM computationally more feasible. The transition probabilities are easily calculated. However, for calculating the emission probabilities, we still need to calculate \(f(g)\) for all \(g \in S\), where \(|S|\) grows exponentially with \(n_{max}\) as mentioned earlier. To overcome this issue, we estimate the emission probabilities using an HMM estimation algorithm such as the Baum–Welch algorithm \cite{rabiner1989tutorial} and do not calculate them directly. In other words, we take the following steps:

\begin{enumerate}
    \item Build our random ERNEC.
    \item Build our NPC from our ERNEC.
    \item Create the \(n_{max}\) by \(n_{max}\) transition probability matrix for transition between different numbers of nodes in our ERNEC (Equation \ref{Matrice}).
    \item Use an HMM estimation algorithm to estimate the emission probabilities of the HMM \((X_i,Y_i)\).
    \item Run the forward algorithm to calculate the probability of the observed sequence.
    \item Estimate the entropy rate of the NPC using the AEP.

\end{enumerate}

The above algorithm will help us estimate the probability of an observed sequence of an NPC in an efficient manner. We can first run the described algorithm on several instances created from the same NPC to have an estimate of its entropy rate and then use that value to indicate if a sequence is typical or not.

\section{Simulation}

For observing the results of our model, we simulated Erdős–Rényi Network Evolution Chains using a MATLAB program. We will show that the results of the simulation agree with the theorems stated in this paper.

\subsection{Overview of the Code}

The maximum number of nodes in the graph (\(n_{max}\)), the Erdős–Rényi parameter, and the length of the random chain to be generated are specified by the user. $t(i)$, $r(i)$, and $s(i)$ are chosen randomly in the following manner. For each \(i\) (\(1<i<n_{max}\)), two points are chosen uniformly in \((0,1)\). These two points divide \([0,1]\) into three areas which will provide us with three positive numbers whose sum equals one. We assign the first one to \(t(i)\), the second one to \(r(i)\) and the third one to \(s(i)\). For \(i=1\), we divide \([0,1]\) into two sections and assign the first one to \(t(1)\) and the second one to \(s(1)\). For \(i=n_{max}\), we divide \([0,1]\) into two sections and assign the first one to \(r(n_{max})\) and the second one to \(s(n_{max})\). After these parameters are set, the chain of graphs is randomly created in the manner described in definition \ref{ERNEC Definition}. Please note that in the following simulations, the value for $n_max$ is deliberately chosen to be small for simplicity in presenting the results.

\subsection{Stationary Distribution}

In this section, we compare the stationary distribution of the number of nodes calculated from theorem \ref{stationary} and the one observed in the simulated chain for a number of random chains.

Table \ref{tab1} and table \ref{tab2} clearly show that the observed values for the stationary distribution of nodes matches their corresponding calculated value. The slight difference that can be seen in both tables between the corresponding calculated and observed values can be attributed to the finite length of the chain. Stationarity of the process guarantees that these values tend to be equal to each other as the length of the chain tends to infinity.

\begin{table}[h!]
\caption{\textbf{Simulation results for a random ERNEC: stationary distributions}}
    \centering
    \begin{center}
    \begin{tabular}{ |c||c|c|c|c|c|  }
     \hline
     \multicolumn{6}{|c|}{ERNEC with \(n_{max}=5\), \(q=0.7\) and length\(=100000\)} \\
     \hline
     i&1&2&3&4&5\\
     \hline
     Calculated \(\pi_i\)   & 0.4589    &0.2717&   0.2274&0.0409&0.0011\\
     Observed \(\pi_i\)&   0.4582  & 0.2744   &0.2251&0.0411&0.0013\\
     \hline
    \end{tabular}
    \end{center}
    
    \label{tab1}
\end{table}

\begin{table}[h!]
\caption{\textbf{Simulation results for a random ERNEC: stationary distributions}}
    \centering
    \begin{center}
    \begin{tabular}{ |c||c|c|c|c|  }
     \hline
     \multicolumn{5}{|c|}{ERNEC with \(n_{max}=8\), \(q=0.7\), and length\(=100000\)} \\
     \hline
     i&1&2&3&4\\
     \hline
    Calculated \(\pi_i\)   & 0.0486    &0.0934&   0.1405&0.2743\\
     Observed \(\pi_i\)& 0.0494    &0.0937&   0.1410&0.2732\\
    \hline
     i&5&6&7&8\\
     \hline
    Calculated \(\pi_i\) &0.2564&0.1780&0.0052&0.0035\\
    Observed \(\pi_i\)&0.2547&0.1772&0.0063&0.0047\\
     \hline
    \end{tabular}
    \end{center}
    
    \label{tab2}
\end{table}

\subsection{Stationary Distribution of Graphs}

In this section, we want to observe the most important property of ERNEC; the fact that the stationary distribution of graphs conditioned on their number of nodes matches the Erdős–Rényi distribution. To this end, an ERNEC is defined with an \(n_{max}\) of 3 and Erdős–Rényi parameter of 0.7. An instance of length 1,000,000 is then created from it. We counted the occurrence of 3 specific graphs with 3 nodes and divided those numbers by the number of occurrence of all of the graphs with 3 nodes. This provided us with an estimation of the graphs' conditioned stationary distribution. We also calculated their expected probability of occurrence based on the \(G(3,0.7)\) model. The results can be seen in table \ref{Stationary Simulation}. It can be observed that the estimated probabilities are very close to their expected ER probabilities and therefore the simulation backs up theorem \ref{proof ER}. The small difference between the observed and the calculated probabilities can again be attributed to the finite length of the chain.

\begin{table}[h]
 \caption{\textbf{Stationary Distribution of Graphs}\label{Stationary Simulation}}
  \centering

  \begin{tabular}{ | c | c | c | }
    \hline
    \multicolumn{3}{|c|}{ERNEC with \(n_{max}=3\), \(q=0.7\), and length\(=1000000\)} \\
    \hline
    Graph & Estimated Probability & ER Probability \\ \hline
    \begin{minipage}{.3\textwidth}
      \includegraphics[width=\linewidth]{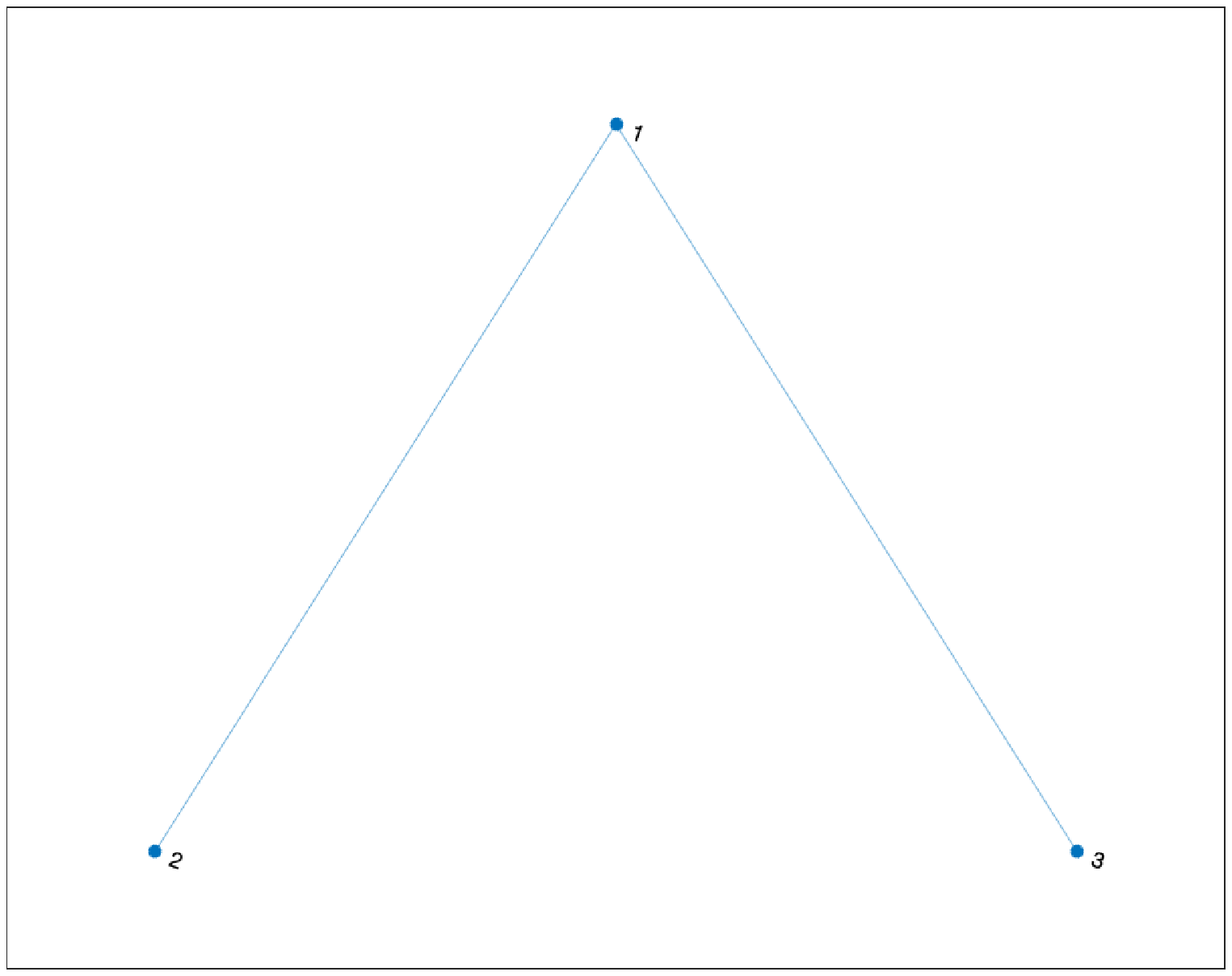}
    \end{minipage}
    &
    0.1453
    & 
    0.1470
    \\ \hline
    \begin{minipage}{.3\textwidth}
      \includegraphics[width=\linewidth]{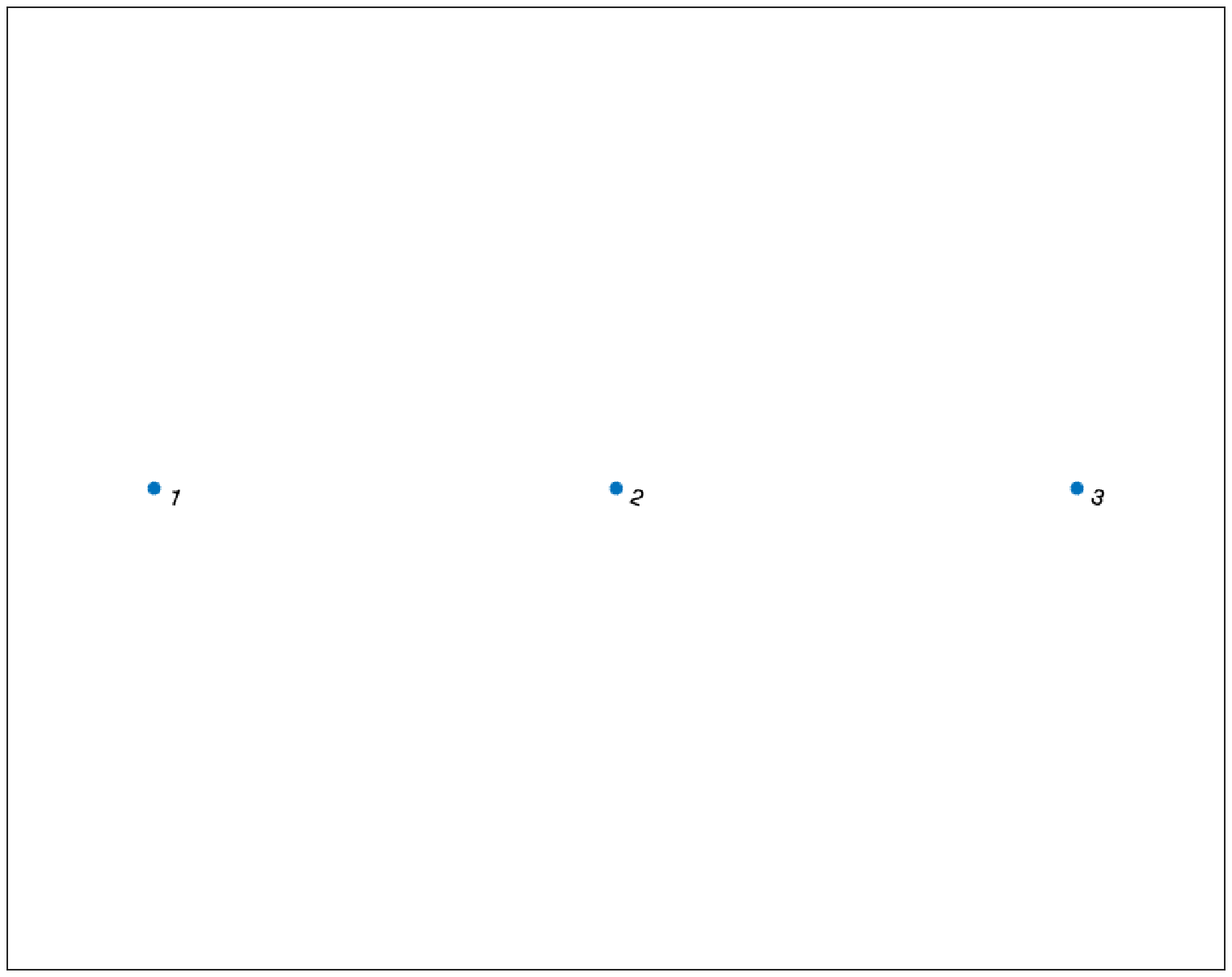}
    \end{minipage}
    &
    0.0263
    & 
    0.0270
    \\ \hline
    \begin{minipage}{.3\textwidth}
      \includegraphics[width=\linewidth]{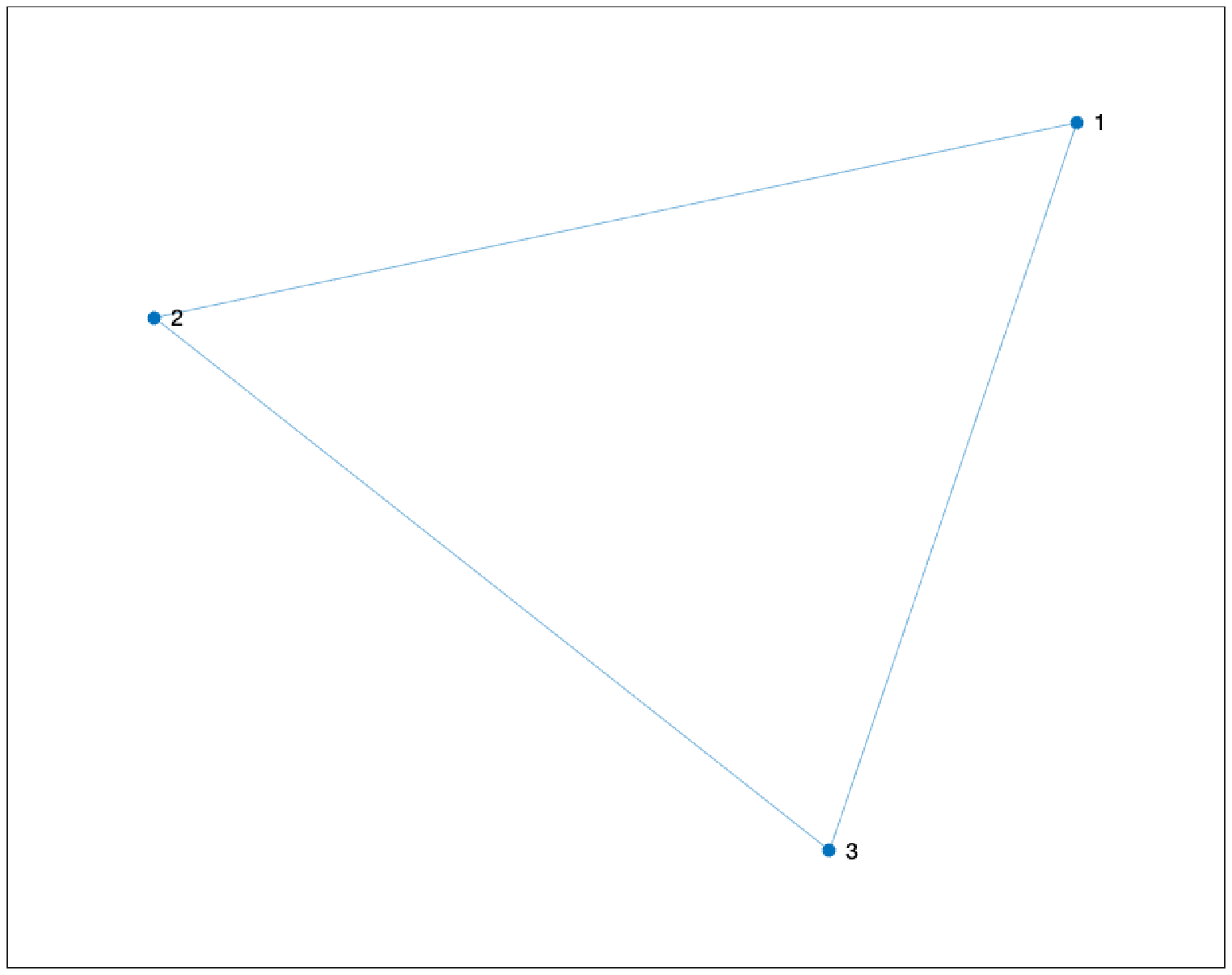}
    \end{minipage}
    &
    0.3372
    & 
    0.3430
    \\ \hline
  \end{tabular}

\end{table}

\subsection{Entropy Rate}

One of the ultimate goals of the concept of Network Evolution Chains is to calculate the entropy rate of the process and be able to define typical sequences of the network. In this section, we examine the properties of entropy rate of ERNEC using two simulations.

For the first simulation, an ERNEC was first defined. Afterwards, its entropy rate was calculated using equation \ref{entropy rate}. Then, five chains of length 200,000 were created from the ERNEC. For each chain, we calculated \(-\frac{1}{n}\log_2p(G_1,...,G_n)\) and then compared it to the calculated entropy rate. The results are shown in table \ref{tab3}.

\begin{table}[h!]
\caption{\textbf{Simulation results for a random ERNEC: Entropy Rate}}
    \centering
    \begin{center}

    \begin{tabular}{ |c|c|c|  }
     \hline
     \multicolumn{3}{|c|}{ERNEC with \(n_{max}=5\), \(q=0.7\), and length\(=200000\)} \\
     \hline
     \multicolumn{2}{|c|}{Calculated Entropy Rate for ERNEC} & 2.6876 b \\
     \hline
     \hline
     instance&\(-\frac{1}{n}\log_2p(G_1,...,G_{50000})\)&\(|H(G)+\frac{1}{n}\log_2p(G_1,...,G_{50000})|\)\\
     \hline
    1&2.7218&0.0342\\
    \hline
    2&2.7180&0.0304\\
    \hline
    3&2.7242&0.0366\\
    \hline
    4&2.7101&0.0225\\
    \hline
    5&2.7082&0.0206\\
    \hline
    \end{tabular}
    
    \end{center}
    
    \label{tab3}
\end{table}

The calculated entropy rate measures the amount of uncertainty we expect when the network goes from one step to the next measure in bits. We expect most of the chains that are randomly created using the model to nearly be a typical sequence. We know that a typical sequence is a sequence whose observed mean entropy is near the actual entropy. It can be seen in table \ref{tab3} that the observed mean entropy for all of the instances is near the actual entropy rate and therefore those instances quite represent some of the typical sequences of length 200000 of graphs that are created using the specified ERNEC.

For the second simulation, we first defined an ERNEC with \(n_{max}=5\) and an Erdős–Rényi parameter of 0.7. Then, we created four instances of length 5000 from it. Fig~\ref{fig2} illustrates the results of this simulation. In each of the subfigures, the estimated entropy rate is plotted as a function of the number of the samples used from the chain to estimate the entropy rate using the AEP. The AEP states that as the length of the chain tends towards infinity, it is expected for the estimated entropy to tend towards the actual entropy rate. This result can be observed for all subfigures in fig~\ref{fig2}, except for \ref{fig2}c. A gap can be observed between the calculated and estimated entropy rate of subfigure \ref{fig2}c. This gap shows that based on the value chosen for $\epsilon$ in definition \ref{Typical Set}, there exist typical sets in which the respective sequences of all subfigures in figure fig~\ref{fig2} are considered to be in the typical set, except for the respective sequence of subfigure \ref{fig2}c. In other words, the respective sequence of subfigure \ref{fig2}c is not as typical a sequence as the other three sequences.

\begin{figure}[!ht]
        \centering
        \includegraphics[width=0.7\columnwidth]{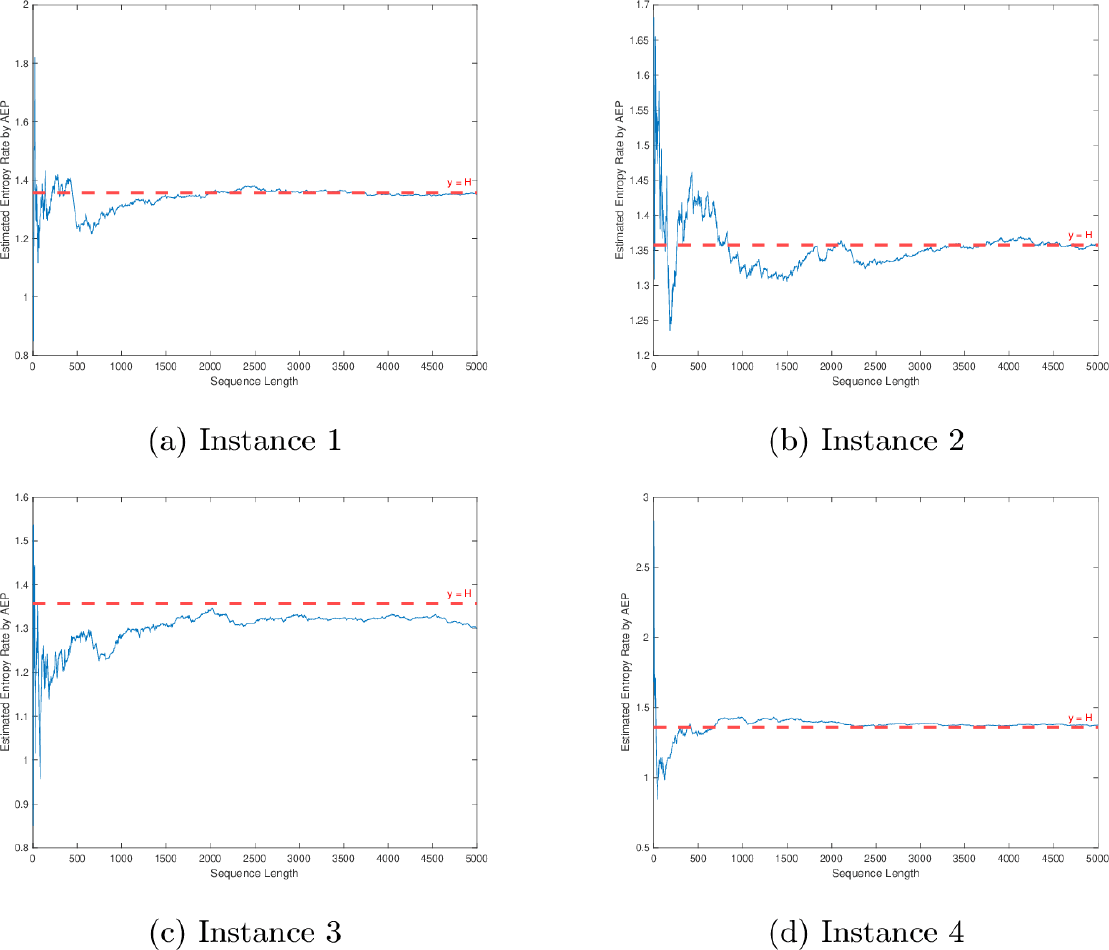}
        \caption{\textbf{The estimated entropy rate for four instances of the same ERNEC}}
        \label{fig2}
\end{figure}

\subsection{Number of Triangles in a Network}

In this section, in addition to simulating an ERNEC using the described method, we also simulate a Network Property Chain. For this simulation, we are interested in the number of triangles formed in the network. Notice that using graph theory, if \(A\) is the adjacency matrix of a network, the number of triangles in the network can be calculated using the following formula:

\[\text{Number of triangles in the network} = \operatorname{Trace}{A^3}/6\]

First, we defined our ERNEC and then built its respective NPC by counting the number of triangles in each network. Afterwards, we formed the HMM defined in section \ref{ER HMM} and then estimated the transition and emission probabilities by having the NPC and the respective states. For this matter, the function \textbf{hmmestimate} in MATLAB was used, which uses a maximum likelihood estimation to estimate the parameters of an HMM \cite{rabiner1989tutorial}. Having the transition and emission probabilities, the forward algorithm calculates the probability of the observed sequence. We can then use the probability of the observed sequence to estimate the entropy rate of the Network Property Chain by the AEP. We chose \(n_{max}=3\). This way, the number of triangles in the network is either zero or one. The Erdős–Rényi probability was set as 0.8 and four instances of length 100000 were created from the defined ERNEC. For each of these four instances, we created its NPC by counting the number of triangles in each graph. We kept the last 1500 samples of each NPC for calculating its entropy rate using the AEP. The results of this simulation are plotted in fig~\ref{fig3}. We have plotted the estimated entropy rate as a function of the number of NPC samples used in equation \ref{AEP to NPC} for estimating the entropy rate. We can observe that, for all of the four plots, the entropy rate is tending towards the same value. Among these four instances, we can see that for fig~\ref{fig3}c the convergence is slower than the others. This can be interpreted as the NPC of fig~\ref{fig3}c is not as typical a sequence as the others. For all of the plots in fig~\ref{fig3}, we can observe that the estimated entropy rate for the NPC is tending towards a value between 0.6 and 0.7. This means that for this specific ERNEC, knowing the number of triangles at time \(n\), there is around 0.7 bits of uncertainty about the number of triangles at time \(n+1\). It also means that the average uncertainty about each symbol in our NPC is around 0.7 bits. Notice that because we only have 2 possible symbols in this NPC, the maximum amount of uncertainty is 1 bit per symbol. However, we can see that having information about the network has reduced this uncertainty to around 0.7 bits.

\begin{figure}[!ht]
        \centering
        \includegraphics[width = 0.7\columnwidth]{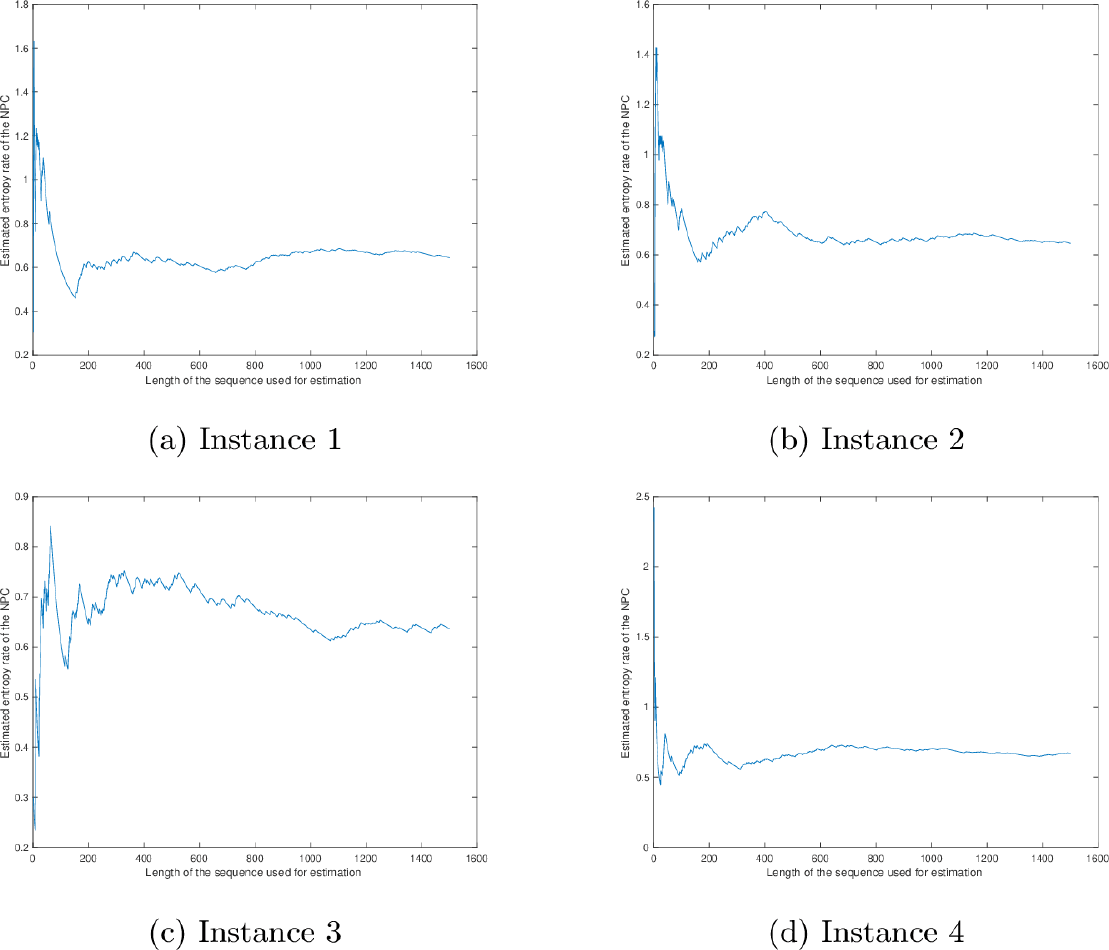}
        \caption{\textbf{The estimated entropy rate for the number of triangles in an NEC}}
        \label{fig3}
\end{figure}

\section*{Conclusion}

Network Evolution Chains were introduced as models which can be used to simulate the evolution of networks over time, while also satisfying the conditions of a stochastic process to which the Shannon-McMillan-Breiman theorem can be applied. This leads to the definition of typical sequences of an NEC, which will act as the basis of compressing the evolution of a network. This will help us in applications in which the state of a network needs to be transmitted in an efficient manner, when the previous states are assumed to be known. Knowing the entropy rate of the process can lead to better compression, transmission and even prediction of the future states. Sometimes, we are interested in certain properties of the network, such as the occurrence of a certain structure. To this end, Network Property Chains were introduced and analysed as processes that only focus on our desired properties. Additionally, two methods were introduced for simulating the concept of the NEC in the continuous time domain, one of which is based on continuous time Markov chains. Furthermore, Erdős–Rényi Network Evolution Chains were defined as a subset of NECs in a way that their conditional distribution matches that of the Erdős–Rényi random graph generation model and can be used in scenarios in which the Erdős–Rényi model is used for simulating networks.

\section*{Acknowledgments}
This work was supported by EPSRC grant number EP/T02612X/1. We also thank Moogsoft inc. for their support in this research project.

%
%
%

\bibliographystyle{comnet}
\bibliography{biblio}

\end{document}